\documentclass[a4paper,11pt]{article}

\usepackage{complexity}
\usepackage{framed}
\usepackage{amsmath,amsthm,amssymb}
\usepackage[auth-sc]{authblk}
\usepackage{graphicx}
\usepackage[T1]{fontenc}
\usepackage{fullpage}

\DeclareMathOperator{\reach}{reach}
\DeclareMathOperator{\true}{TRUE}
\DeclareMathOperator{\false}{FALSE}
\DeclareMathOperator{\maxtime}{maxtime}
\DeclareMathOperator{\out}{out}
\DeclareMathOperator{\opt}{opt}

\newcommand{\tempOrd}{\textsc{Min-Max Reachability Temporal Ordering}}

\newcommand{\singTempOrd}{\textsc{Singleton} \tempOrd}

\newtheorem{obs}{Observation}
\newtheorem*{obs*}{Observation}
\newtheorem*{theorem*}{Theorem}
\newtheorem*{lemma*}{Lemma}
\newtheorem*{proposition*}{Proposition}
\newtheorem{prop}{Proposition}
\newtheorem{thm}{Theorem}
\newtheorem{lma}{Lemma}
\newtheorem{cor}{Corollary}
\newtheorem{claim}{Claim}
\newtheorem{conjecture}{Conjecture}
\newtheorem{openproblem}{Open Problem}

\title{Assigning times to minimise reachability in temporal graphs}

\date{August 2020}

\author[1]{Jessica Enright\thanks{Partially supported by EPSRC Project EP/P026842/1: Modelling and Optimisation with Graphs.}}
\author[2]{Kitty Meeks\thanks{Supported by a Royal Society of Edinburgh Personal Research Fellowship, funded by the Scottish Government.}}
\author[3]{Fiona Skerman\thanks{Supported by grants from Swedish Research Council and the Ragnar S\"oderberg Foundation.}}
\affil[1]{\small{Global Academy of Agriculture and Food Security, University of Edinburgh, Edinburgh, UK\\ \textit{Current address:} School of Computing Science, University of Glasgow, Glasgow, UK \\ \textit{Email:} \texttt{jessica.enright@glasgow.ac.uk} \vspace{6pt}}}
\affil[2]{School of Computing Science, University of Glasgow, Glasgow, UK\\ \textit{Email:} \texttt{kitty.meeks@glasgow.ac.uk} \vspace{6pt}}
\affil[3]{Department of Mathematics, Uppsala University, Uppsala, Sweden\\ \textit{Email:} \texttt{skerman@fi.muni.cz} }

\begin{document}

\maketitle

\begin{abstract}
Temporal graphs (in which edges are active at specified times) are of particular relevance for spreading processes on graphs, e.g.~the spread of disease or dissemination of information.   Motivated by real-world applications, modification of static graphs to control this spread has proven a rich topic for previous research.

Here, we introduce a new type of modification for temporal graphs: the number of active times for each edge is fixed, but we can change the relative order in which (sets of) edges are active.  

We investigate the problem of determining an ordering of edges that minimises the maximum number of vertices reachable from any single starting vertex; epidemiologically, this corresponds to the worst-case number of vertices infected in a single disease outbreak.

We study two versions of this problem, both of which we show to be $\NP$-hard, and identify cases in which the problem can be solved or approximated efficiently.
\end{abstract}

\section{Introduction}

Temporal (or dynamic) graphs have emerged recently as a useful structure for representing real-world situations, and as a rich source of new algorithmic problems \cite{akrida17,himmel17,temporalReview,kempe02,liang17,mertzios13,temporalReview1}.  A temporal graph changes over time: each edge in the graph is only \emph{active} at certain timesteps.  Some notions from the study of static graphs transfer immediately to the temporal setting, but others -- including the very basic notions of connectivity and reachability -- become much more complex in the temporal setting.

Temporal graphs are therefore of particular interest when considering the dynamics of spreading processes on graphs, for example the spread of a disease or the dissemination of sensitive information.  The maximum number of vertices that can be reached from a single starting vertex $v$, known as the \emph{reachability} of $v$, has emerged as an important measure in both epidemiology \cite{edge-deletion,temp-edge-del,li11} and the study of network vulnerability \cite{drange14,gross13}; note that in the case of static, undirected graphs the reachability of $v$ is equal to the number of vertices in the connected component containing $v$.

Previous work on this topic has addressed questions of the following form: given a set of rules for how an input graph $G$ can be modified, how small can we make the maximum reachability taken over all vertices of $G$?  The answer to these questions can be interpreted as the effectiveness of an optimal intervention to control a spreading process on the graph, or the worst-case impact of an attack on a graph's reliability.  Modification rules considered in the literature include vertex deletions, edge deletions, and the deletion of edges at some budgeted number of times.  The related optimisation problems of deleting edge-times to minimise various cost measures while ensuring that the set of vertices reached from every vertex remains unchanged have also been studied \cite{akrida17,mertzios13} (the notion of ``assignment'' of times to edges introduced in \cite{mertzios13} can be regarded as an edge-time deletion problem when all edges are initially active at all times).

Graphs of livestock trades present us with an example in which it is desireable to reduce the maximum reachability of a graph, as discussed in \cite{edge-deletion,temp-edge-del}.  Here, vertices represent farms, and the sale of an animal from one farm to another (together with the associated risk of disease transmission) is naturally represented with a (directed) edge.  If a disease incursion starts at the farm represented by $v$, then the farms at risk of infection are precisely those in the reachability set of $v$, and it is natural to try to minimise the worst case number of farms that might be infected.

It is not clear, however, how one might realistically remove edges from such a graph: forbidding trade is likely impossible.   A more realistic hope is that the set of edges (e.g.~the trades that take place) is fixed, but we can intervene to alter the relative times at which the edges are active (i.e.~reorder the trades, via alterations in sale and auction dates).

The focus of this paper is on precisely this kind of temporal graph modification: we cannot remove edges, but we can choose the order in which they are active.  (In contrast with the model of \cite{mertzios13}, the number of edge-times to be assigned to each edge is fixed.)  In the simplest version of the problem, we can reorder edges with complete freedom, but for many applications there are likely to be additional constraints which stop us rescheduling edges independently.  Specifically, we may require that particular subsets of edges are all active simultaneously: this corresponds, for example, to a set of contacts that will take place at a particular conference, or the trades that will be made at a specific named livestock market event (e.g.~the ``Spring Bull Sale''), whenever the event in question is scheduled.  Indeed, scenarios of this kind where the timing of contacts can be controlled by an organisation responsible for scheduling events (e.g.~an auctioneer) perhaps represent the most likely real-world application for this rescheduling approach. We therefore consider a more general version of the problem involving assigning times to classes of simultaneously occurring edges.  We mostly see that both versions of the problem are intractable, but we are able to identify a small number of special cases which admit polynomial-time exact or approximation algorithms.  Where possible, the restricted graph classes we consider in our pursuit of tractability are either inspired by our epidemiological motivation, or are a first step in the direction of such classes.

The rest of the paper is organised as follows.  We describe our model in more detail in Section \ref{sec:notation}, give formal definitions of the problems we consider in Section \ref{sec:probdefs}, and summarise our results in Section \ref{sec:summary}.  Section \ref{sec:singleton} is devoted to the setting in which edges can be reordered independently, and in Section \ref{sec:general} we consider the generalisation of the problem involving classes of simultaneously active edges.

\subsection{Our model, notation, and some simple observations on reachability}
\label{sec:notation}

A vertex $v$ is \emph{reachable} from vertex $u$ in a (di)graph if there is a (directed) path from $u$ to $v$ in that (di)graph.  The \emph{reachability set} of a vertex in a (di)graph is the set of all vertices reachable from that vertex.  A \emph{directed acyclic graph (DAG)} is a directed graph that does not contain any directed cycle.  An edge is \emph{incident} at the two vertices that it contains, and we say that two edges are \emph{incident} if they share a vertex. In any graph, we refer to a vertex of degree one as a \emph{leaf}, and any edge incident with a leaf as a \emph{leaf edge}.  For any natural number $n$, we write $[n]$ as shorthand for the set $\{1,\ldots,n\}$.  

In its most general form, a \emph{temporal} graph is a static (directed) graph $G = (V, E)$, which we refer to as the \emph{underlying graph}, together with a function $\mathcal{T}$ that maps each edge to a list of timesteps at which that edge is active.  Note that these timesteps need not give a continuous interval for any particular edge.  A \emph{strict temporal path} (sometimes called a \emph{time-respecting path} in the literature) from $u$ to $v$ in a temporal graph $(G=(V,E),\mathcal{T})$ is a (directed) path from $u$ to $v$ composed of edges $e_0,e_1...e_k$ such that each edge $e_i$ is assigned a time $t(e_i)$ from its image in $\mathcal{T}$ where $t(e_i)< t(e_{i+1})$ for $0 \leq i < k$.  We say a vertex $v$ is \emph{temporally reachable} from $u$ in a temporal graph $(G = (V, E),\mathcal{T})$, if there is a strict temporal path from $u$ to $v$; we adopt the convention that every vertex $v$ is temporally reachable from itself.  A vertex $v$ is temporally reachable from $u$ \emph{after time $\tau$} if there is a temporal path from $u$ to $v$ in which every edge is assigned a time strictly greater than $\tau$.  A \emph{weak temporal path}, and the concept of \emph{weak temporal reachability} are defined analogously, where we replace the requirement $t(e_i) < t(e_{i+1})$ with a weak inequality.  Both notions have been considered in the literature on temporal graphs (see, for example \cite{kempe02,mertzios13,temporalReview1,zschoche17}).  For simplicity, we use the notion of strict temporal reachability throughout this paper (and, unless otherwise stated, all temporal paths should be assumed to be strict); however, all of our results also hold in the setting of weak reachability, since there is always an optimal assignment of times to edges under weak reachability in which every edge (or edge class) is assigned a different time.  

We call the set of vertices that are temporally reachable from vertex $u$ in a temporal graph the \emph{temporal reachability set of $u$}.  Note that the temporal reachability set of any vertex $u$ in the temporal graph $(G,\mathcal{T})$ is a subset of its reachability set in the static underlying graph $G$.  The \emph{maximum temporal reachability} of a temporal graph is the maximum cardinality of the temporal reachability set of any vertex in the graph.  Note that the temporal reachability set of any vertex can be computed using a breadth-first search, and so the maximum reachability can be computed in polynomial time by considering each vertex in turn.  

As we are only interested in reachability in the sense defined above, we do not care about the absolute times at which edges are active, simply the order in which they are active.  Moreover, we will assume that it is determined in advance which subsets of edges will be active simultaneously (in the simpler version of the model, each such subset is a singleton), and the number of timesteps at which each such subset is active: this can be represented by a multiset $\mathcal{E}=\{E_1,\ldots,E_h\}$ of subsets of the edge-set of $G$, where each element of $\mathcal{E}$ is a subset of edges which is active simultaneously, and the multiplicity of each element in $\mathcal{E}$ is equal to the number of timesteps at which the subset is active.  This model is appropriate when each element of $\mathcal{E}$ corresponds to the connections that will be active when a particular event (e.g.~a livestock market, a conference, or a flight) takes place: we can potentially change the relative timing of the various events, but each event will cause a fixed subset of edges to be active.

To avoid the problem becoming trivial, we further restrict the way in which we may alter timing by requiring that each element of $\mathcal{E}$ is assigned a distinct timestep: without this restriction, we could always minimise the maximum reachability by assigning all edge subsets the same timestep.  We therefore assume throughout that every assignment of times to edges is a bijection (and we shall use the terms assignment and bijection interchangeably in this setting).  The timesteps at which each edge is active can thus be specified by the multiset $\mathcal{E}$ together with a bijection $t \colon \mathcal{E} \rightarrow [|\mathcal{E}|]$: the edge $e$ is active at time $s$ if and only if $t(E_i) = s$ for some $E_i \in \mathcal{E}$ with $e \in E_i$.  We will write $(G,\mathcal{E},t)$ for a temporal graph represented in this way, and $\reach_{G,\mathcal{E},t}(u)$ for the temporal reachability set of $u$ in $(G,\mathcal{E},t)$.

We now make some simple observations about reachability sets. 

\begin{obs}\label{prop:leaf}
Let $G$ be an undirected graph, $v$ a vertex of degree one in $G$, and $u$ the unique neighbour of $v$.  Then $\reach_{G,\mathcal{E},t}(v) \subseteq \reach_{G,\mathcal{E},t}(u)$.
\end{obs}

\begin{obs}\label{prop:first-edge}
Let $uv$ be an edge in the undirected graph $G$.  Suppose that $uv$ is first active at time $t_1$ and that no other edge incident with $v$ is active at or before time $t_1$.  Then $\reach_{G,\mathcal{E},t}(v) \subseteq \reach_{G,\mathcal{E},t}(u)$.
\end{obs}

For the next observation, we need one more piece of notation: we write $N_G(v)$ for the set of vertices in $G$ that are adjacent to $v$.

\begin{obs}\label{prop:degree}
For any undirected temporal graph $(G,\mathcal{E},t)$ and $v \in V(G)$, we have $\reach_{G,\mathcal{E},t}(v) \supseteq N_G(v) \cup \{v\}$.
\end{obs}

In a directed graph $G$, we write $N_{G}^{out}(v)$ for the set of vertices $\{u: \overrightarrow{vu} \in E(G)\}$.  We can now make an analogous observation for the directed case.

\begin{obs}\label{prop:out-degree}
For any directed temporal graph $(G,\mathcal{E},t)$ and $v \in V(G)$, we have $\reach_{G,\mathcal{E},t}(v) \supseteq N_G^{out}(v) \cup \{v\}$.
\end{obs}

Finally, we define the \emph{edge-class interaction graph of $(G,\mathcal{E})$} (where $G$ is an undirected graph) to be the graph with vertex-set $[h]$ in which vertices $i$ and $j$ are adjacent if and only if there exist $e_i \in E_i$ and $e_j \in E_j$ such that $e_i$ and $e_j$ are incident.  

\begin{obs}\label{prop:independent}
Suppose that $i$ and $j$ are non-adjacent vertices in the edge-class interaction graph of $(G,\mathcal{E})$, and that the sets $E_i$ and $E_j$ are active at consecutive timesteps.  Then swapping the timesteps assigned to $E_i$ and $E_j$ does not change the reachability set of any vertex in $G$.
\end{obs}

An analogous result holds for directed graphs if we define the edge-class interaction graph by making $i$ and $j$ adjacent if and only if there is some $v \in V(G)$ such that $v$ has an in-edge in $E_i$ and an out-edge in $E_j$, or vice versa.

\subsection{Problems considered}
\label{sec:probdefs}

The main problem we consider in this paper is the following.

\begin{framed}
\noindent
\textbf{\textsc{Min-Max Reachability Temporal Ordering}}\\
\textit{Input:} A graph $G = (V, E)$, a list $\mathcal{E}= \{E_1,\ldots,E_h\}$ of subsets of $E$, and a positive integer $k$.\\
\textit{Question:} Is there a bijective function $t \colon \mathcal{E} \rightarrow [h]$ such that the maximum temporal reachability of $(G,\mathcal{E},t)$ is at most $k$?
\end{framed}

The simplest case is when $\mathcal{E}$ consists of pairwise disjoint singleton sets (so that all edges can be reordered independently; this is analogous to the notion of single-edge single-label temporal graphs in \cite{akrida17}).  We refer to this special case as \textsc{Singleton Min-Max Reachability Temporal Ordering}.  In this singleton setting, we will sometimes abuse notation and consider $\mathcal{E}$ to be equal to the edge set $E$ of $G$ (rather than the set of all singleton sets containing one edge of $E$), and describe a temporal graph as $(G,E,t)$ where $E$ is the edge-set of $G$.  In this setting we may assume without loss of generality that $\mathcal{E}$ is a set: if an edge must be active at more than one timestep, we can always minimise reachability by assigning it consecutive timesteps, in which case the reachabilities are the same as if each edge is active at a single timestep.

Observe that changing the order in which edges are active can have a huge impact on the reachability of vertices in the graph.  For example, let $G$ consist of a clique on $r$ vertices, where each vertex in the clique has $s$ pendant leaves.  If all edges in the clique are active (in an arbitrary order) before all edges incident with the leaves, then each vertex in the clique will reach the entire graph, a total of $r(s+1)$ vertices.  On the other hand, if edges incident with leaves are active first (in some order), then no vertex reaches more than $r+s$ vertices (including itself).

We now argue that our problems belong to $\NP$.

\begin{prop}\label{prop:compute}
We can compute the maximum temporal reachability of a temporal graph $(G,\mathcal{E},t)$ in polynomial time.
\end{prop}
\begin{proof}
For any vertex $v \in G$, we can compute $reach_{G,\mathcal{E},t}(v)$ in polynomial time as in \cite{kempe02,mertzios13}.  Therefore, considering each vertex in turn, we can find the maximum temporal reachability of the whole temporal graph in polynomial time.

\end{proof}

\begin{cor}\label{cor:inNP}
\tempOrd\, belongs to $\NP$.
\end{cor}

We conclude this section with two simple observations about situations in which both problems admit efficient algorithms; here $\Delta(G)$ denotes the maximum degree of the graph $G$.

\begin{prop}\label{prop:trivial-kdeg}
\tempOrd\, is solvable in polynomial-time when $k \leq \Delta(G)$.
\end{prop}
\begin{proof}
By Observation \ref{prop:degree}, we know that the maximum reachability set of $(G,\mathcal{E},t)$ is at least $\Delta(G) + 1 > k$, so we must have a no-instance.
\end{proof}

\begin{prop}
\tempOrd\, belongs to $\FPT$ when parameterised by the number $h$ of timesteps.
\end{prop}
\begin{proof}
Given any instance of \tempOrd, there are at most $h!$ possible orderings of the edge subsets, as each must be a bijective function from $\{E_1,\ldots,E_h\}$ to $h$.  Having fixed such an ordering we can, by Proposition \ref{prop:compute}, find the maximum reachability of the corresponding temporal graph in polynomial time.  It follows that we can solve \tempOrd\, in polynomial time by considering each of the $h!$ orderings in turn.   
\end{proof}

\subsection{Summary of results}
\label{sec:summary}

We begin by considering the special case in which $\mathcal{E}$ consists of pairwise disjoint, singleton edge sets.  We show that \singTempOrd\, is $\NP$-complete on general graphs, but that the optimisation version admits a linear-time constant-factor approximation algorithm when restricted to graphs of bounded maximum degree.  We also consider the problem on some special classes of graphs, and show that \singTempOrd\, is polynomial-time solvable on DAGs and, when restricted to trees, admits an FPT algorithm parameterised by either the maximum permitted reachability $k$ or the maximum degree; we also give an improved approximation algorithm for trees.

In the more general case of \tempOrd, we show that the problem remains $\NP$-complete even on trees and DAGs, and moreover is $\W[1]$-hard on trees when parameterised by the vertex cover number.  In this setting we can obtain a constant-factor approximation to the optimisation problem on graphs of bounded maximum degree provided that the edge-class interaction graph is bipartite or has bounded degree.

\section{The case of singleton edge classes}
\label{sec:singleton}

In this section we consider the restricted version of the problem in which we can reschedule each edge independently.  In Section \ref{sec:singleton-hard} we show that even this special case of the problem is $\NP$-complete on general graphs, before deriving some general bounds on the achievable maximum reachability which lead to an approximation algorithm for bounded degree graphs in Section \ref{sec:singleton-approx}.  We then give efficient exact algorithms for some special cases in Section \ref{sec:singleton-special}.

Recall that we abuse notation in the singleton case, frequently referring to a temporal graph as $(G,E,t)$ where $E$ is the edge-set, rather than $\left(G, \mathcal{E} = \left\lbrace \{e\} \right\rbrace, t\right)$, and considering the domain of $t$ to be $E$ rather than $\mathcal{E}$ (so $t$ is an assignment of times to edges rather than to singleton sets of edges).

Several results in this section will make use of the following simple lemma.

\begin{lma}\label{lma:leaves-first}
Let $(G,E,k)$ be a yes-instance of \singTempOrd.  Then there is an assignment $t: E \rightarrow [|E|]$ of times to edges such that the maximum reachability of $(G,E,t)$ is at most $k$ and all leaf edges are assigned times strictly before all other edges.  Moreover, in such an assignment, the relative order of the leaf edges does not change the maximum reachability.
\end{lma}
\begin{proof}
Suppose that $t: E \rightarrow [|E|]$ is an assignment of times to edges such that the maximum reachability of $(G,\mathcal{E},t)$ is at most $k$, and suppose that $v$ is a leaf with neighbour $u$.  We claim that, if $t'$ is the assignment of times to edges obtained from $t$ by moving $uv$ to the start and leaving the relative order of edges otherwise unchanged, then no vertex in $G$ has reachability set of size more than $k$ in $(G,E,t')$.

To see this, first note that, for any vertex $w \notin \{u,v\}$, we have that $v \notin \reach_{G,E,t'}(w)$ and $\reach_{G,E,t}(w) \setminus \{v\} = \reach_{G,E,t'}(w)$.
Thus the reachability set of $w$ can only decrease in size, so by assumption $|\reach_{G,E,t'}(w)| \leq k$.  Next, observe that the reachability set of $u$ is unchanged as it includes $v$ under any ordering, and so has size at most $k$. Finally we observe that, although the reachability set of $v$ could increase in size, we must (by Observation \ref{prop:leaf}) have $\reach_{(G,E,t')}(v) \subseteq \reach_{(G,E,t')}(u)$ and so, by the reasoning above, $|\reach_{(G,E,t')}(v)| \leq k$.

Finally, to see that the relative ordering of leaf edges does not change the maximum reachability, we consider two cases.  First, if the leaf edges $e_1$ and $e_2$ are not incident, we know by Observation \ref{prop:independent} that swapping the times assigned to $e_1$ and $e_2$ has no effect on the reachability of any vertex.  Secondly, if $e_1$ and $e_2$ are incident, then swapping their assigned times results in an isomorphic temporal graph (i.e.~the vertices can be relabelled to give an identical temporal graph) and so in particular the maximum reachability does not change.
\end{proof}

\subsection{$\NP$-completeness}\label{sec:singleton-hard}

In this section we prove that \singTempOrd\, is $\NP$-complete.  

\begin{thm}\label{thm:sing-hard-new}
\singTempOrd\, is $\NP$-complete even when the maximum degree $\Delta(G)$ of the input graph and the maximum permitted reachability $k$ are bounded by constants.
\end{thm}
By Corollary \ref{cor:inNP}, (\textsc{Singleton}) \tempOrd\, is in $\NP$, so it remains to show that this version of the problem is $\NP$-hard. 
We prove this by means of a reduction from the following problem, shown to be $\NP$-complete in \cite{berman04}.
\begin{framed}
\noindent
\textbf{\textsc{$(3,2B)$-SAT}}\\
\textit{Input:} A CNF formula $\Phi$ in which every clause contains three literals, and every variable appears positively in exactly two clauses and negatively in exactly two clauses.\\
\textit{Question:} Is $\Phi$ satisfiable?
\end{framed}
Suppose our instance of (3,2B)-\textsc{SAT} is the formula $\Phi = C_1 \wedge \cdots \wedge C_m$, with variables $x_1,\ldots,x_n$; we will assume without loss of generality that the literals appearing in each clause are all distinct.  We will construct an instance $(G,E,k)$ of \singTempOrd\, which is a yes-instance if and only if $\Phi$ has a satisfying assignment.

We begin by describing the construction of the graph $G$.  For each clause $C_j$ in $\Phi$ we have a vertex $v_j$, and for each variable $x_i$ we have a three-vertex path on vertices $t_i$, $w_i$ and $f_i$ (with edges $t_iw_i$ and $w_if_i$).  For each occurrence of a positive literal $x_i$ in a clause $C_j$ we add a new \emph{subdividing} vertex adjacent to $v_j$ and $t_i$, and for each negative literal $\neg x_{\ell}$ in the $C_j$ we add a new subdividing vertex adjacent to $v_j$ and $f_{\ell}$; note that, as every variable appears twice positively and twice negatively in $\Phi$, every vertex $t_i$ or $f_i$ is adjacent to exactly two of these subdividing vertices.  Additionally, for each $i$, we add $21$ leaves adjacent to $w_i$ and $19$ leaves adjacent to each of $t_i$ and $f_i$.  We complete the construction by adding three additional neighbours $a_{i,1}$, $a_{i,2}$ and $a_{i,3}$ to each $t_i$, and symmetrically adding three additional neighbours $b_{i,1}$, $b_{i,2}$ and $b_{i,3}$ to $f_i$; each vertex $a_{i,\ell}$ and $b_{i,\ell}$ itself has $(29-\ell)$ additional leaf neighbours.  The construction of $G$ is illustrated in Figure \ref{fig:bdd_deg_hardness}.  Our instance of \singTempOrd\, is $(G,E(G),30)$.

\begin{figure}
\centering
\includegraphics[width = 0.75 \linewidth]{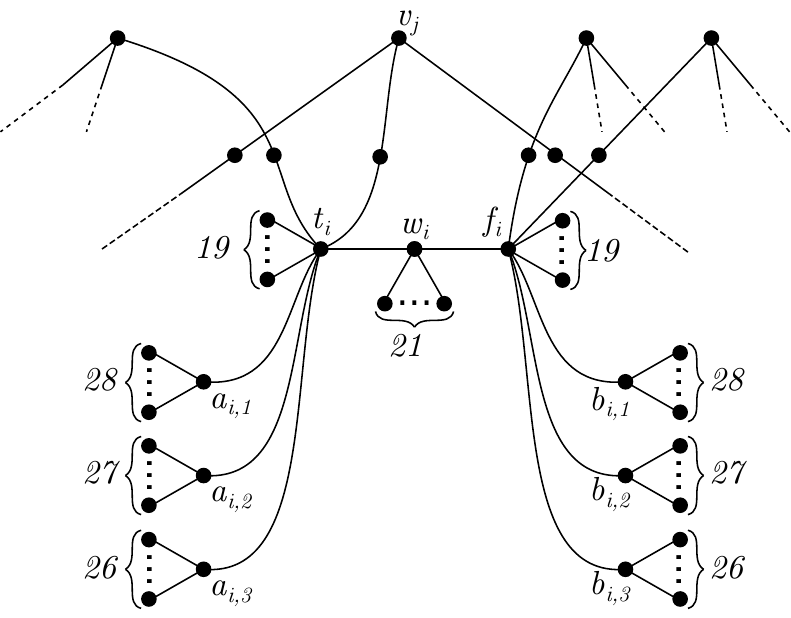}
\caption{Part of the construction of the graph $G$; in this example variable $x_i$ appears positively in clause $C_j$.}
\label{fig:bdd_deg_hardness}
\end{figure}

\begin{lma}\label{lma:truth->order}
Suppose that $\Phi$ has a satisfying truth assignment $b: \{x_1,\ldots,x_n\} \rightarrow \{\true,\false\}$.  There exists an ordering of the edges of $G$ under which no reachability set has cardinality greater than $30$.
\end{lma}
\begin{proof}
We describe an ordering of the edges of $G$ derived from the truth assignment $b$; this order consists of the following edges, in order:
\begin{enumerate}
\item all edges incident with a leaf, in an arbitrary order;
\item all edges incident with any clause-vertex $v_j$, in an arbitrary order;
\item all remaining edges on any two-edge path from a clause-vertex to a variable-vertex $t_i$ or $f_i$, in an arbitrary order;
\item the edges $t_iw_i$ for every $i$ such that $b(x_i) = \false$, in arbitrary order;
\item the edges $f_iw_i$ for every $i$ such that $b(x_i) = \true$, in arbitrary order;
\item the edges $t_iw_i$ for every $i$ such that $b(x_i) = \true$, in arbitrary order;
\item the edges $f_iw_i$ for every $i$ such that $b(x_i) = \false$, in arbitrary order;
\item for each $i$, the edges $a_{i,3}, a_{i,2}, a_{i,1}, b_{i,3}, b_{i,2}, b_{i,1}$, in that order.
\end{enumerate}
Notice that $t_iw_i$ is active before $f_iw_i$ if and only if $b(x_i) = \false$.

First note that, by Observation \ref{prop:leaf}, the reachability set of any leaf vertex is a subset of the reachability set of its unique neighbour, therefore we do not need to bound the cardinality of the reachability set of any leaf.  Moreover, under the ordering described above, the middle vertex on any of the two-edge paths from a clause-vertex $v_j$ to a vertex $t_i$ or $f_i$ reaches a subset of the vertices reached by $v_j$ itself, so we do not need to consider these vertices.  It therefore suffices to show that none of $v_j$, $t_i$, $f_i$, $a_{i,\ell}$ or $b_{i,\ell}$ (for any value of $j$, $i$, or $\ell$) reaches more than $30$ vertices.

We first consider a vertex $a_{i,\ell}$ (the case for $b_{i,\ell}$ follows by the same argument).  We claim that the reachability set of $a_{i,\ell}$ contains precisely its neighbours together with the (possibly empty) set of vertices $\{a_{i,\ell'}: \ell' < \ell\}$: to see this, observe that every edge incident with any of $\{a_{i,\ell}: 1 \leq \ell \leq 3\}$ or $t_i$ is active strictly before the edge between a vertex $a_{i,\ell}$ (for any $\ell$) and $t_i$.  It therefore follows that the reachability set of $a_{i,\ell}$ has cardinality $1 + (29 - \ell) + 1 + (\ell - 1) = 30$, as required.

Now consider a vertex $w_i$.  Under the specified ordering, it is easy to verify that $w_i$ reaches precisely itself and its pendant leaves, together with $t_i$, $a_{i,1}$, $a_{i,2}$, $a_{i,3}$, $f_i$, $b_{i,1}$, $b_{i,2}$ and $b_{i,3}$.  This set has cardinality $1 + 21 + 8 = 30$, as required.

Next consider a vertex $t_i$ (again, the case for $f_i$ will follow by a symmetric argument).  Under the specified ordering, $t_i$ reaches itself, its own neighbours (consisting of $19$ pendant leaves, $a_{i,1}$, $a_{i,2}$ and $a_{i,3}$, its two neighbours on clause-vertex paths, and $w_i$); if $b(x_i) = \false$ then additionally it reaches $f_i$, $b_{i,1}$, $b_{i,2}$ and $b_{i,3}$.  Thus its reachability set has cardinality either $26$ or $30$, depending on the truth value of $x_i$, and in either case we meet the requirement.

To complete the proof of the lemma, it remains only to demonstrate that the reachability set of each clause vertex $v_j$ is not too large.  Suppose that the corresponding clause $C_j$ is $(\ell_1 \vee \ell_2 \vee \ell_3)$, where $\ell_r \in \{x_{i_r}, \neg x_{i_r}\}$; we write $u_r$ for the vertex corresponding to $ell_r$ (so $u_r = t_{i_r}$ if $\ell_r = x_{i_r}$, and $u_r = f_{i_r}$ if $\ell_r = \neg x_{i_r}$).  It is clear that, under the given ordering, $v_j$ reaches itself and every vertex at distance at most two from $v_j$; this set consists of seven vertices, including one $u_1$, $u_2$ and $u_3$.  We further observe that, for each $r \in \{1,2,3\}$, $v_j$ reaches $w_{i_r}$ and either the set $\{a_{i_r,1},a_{i_r,2},a_{i_r,3}\}$ or $\{b_{i_r,1},b_{i_r,2},b_{i_r,3}\}$, as well as possibly one vertex on the two-edge path from $u_r$ to another clause-vertex.  Moreover, if $b(\ell_i) = \false$ then $v_j$ will in fact reach \emph{both} of $t_{i_r}$ and $f_{i_r}$, as well as all of $\{a_{i_r,1},a_{i_r,2},a_{i_r,3}\} \cup \{b_{i_r,1},b_{i_r,2},b_{i_r,3}\}$.  Thus the total number of vertices reached by $v_j$ is at most
$$7 + (3 \times 5) + 4\left|\left\lbrace i \in \{1,2,3\}: b(\ell_i) = \false\right\rbrace\right|.$$
Since, by assumption, $b$ is a satisfying assignment for $\Phi$, we know that at most two of the literals in $C_j$ evaluate to FALSE under $b$, so we conclude that the reachability set of $v_j$ has cardinality at most $30$.
\end{proof}

\begin{lma}\label{lma:order->truth}
Suppose that there is an assignment $t$ of times to edges such that the maximum temporal reachability of $(G,E,t)$ is at most $30$. Then $\Phi$ has a satisfying assignment.
\end{lma}
\begin{proof}
We define a truth assignment $b: \{x_1,\ldots,x_n\} \rightarrow \{\true,\false\}$ as follows:
\begin{equation*}
b(x_i) = \begin{cases}
				\true	& \text{if $t(f_iw_i) < t(t_iw_i)$} \\
				\false  & \text{otherwise.}
		 \end{cases}
\end{equation*}
We will suppose, for a contradiction, that there is some clause $C_j = (\ell_1 \vee \ell_2 \vee \ell_3)$ in $\Phi$ that is not satisfied under this assignment.  Recall from Lemma \ref{prop:leaf} that we may assume without loss of generality that all edges to leaves occur strictly before all others in the ordering $t$.  We begin by making some further deductions about the relative ordering of edges under $t$.

\begin{claim}\label{claim:a_edges}
For each $i$, all edges in the set $\{a_{i,\ell}t_i: 1 \leq \ell \leq 3\}$ occur strictly after any other edge incident with $t_i$.
\end{claim}
\begin{proof}[Proof of Claim \ref{claim:a_edges}]
Notice first that $a_{i,1}$ has degree $29$; since every vertex necessarily reaches itself and its neighbours, we can deduce from the fact that its reachability set has cardinality at most $30$ that $a_{i,1}$ does not reach any vertex outside this set.  In order for this to happen, it must be that every edge incident with $t_i$ is active strictly after $a_{i,1}t_i$ itself.  It therefore follows that $a_{i,2}$ and $a_{i,3}$ both reach $a_{i,1}$.  Now considering $a_{i,2}$, we observe that $a_{i,2}$ necessarily reaches $30$ vertices (itself, its neighbours and $a_{i,1}$) so, to avoid reaching more than this, we must have that every edge incident with $t_i$ other than $a_{i,1}t_i$ is active strictly after $a_{i,2}t_i$, and in particular $a_{i,3}$ must reach $a_{i,2}$.  Applying the same reasoning to $a_{i,3}$ completes the proof of this claim.
\renewcommand{\qedsymbol}{$\square$ (Claim \ref{claim:a_edges})}
\end{proof}

\begin{claim}\label{claim:b_edges}
For each $i$, all edges in the set $\{b_{i,\ell}f_i: 1 \leq \ell \leq 3\}$ occur strictly after any other edge incident with any edge in the set.
\end{claim}
\begin{proof}[Proof of Claim \ref{claim:b_edges}]
This follows by the same reasoning as Claim \ref{claim:a_edges}.
\renewcommand{\qedsymbol}{$\square$ (Claim \ref{claim:b_edges})}
\end{proof}

\begin{claim}\label{claim:truth_edges}
For each $i$, both $t_iw_i$ and $f_iw_i$ occur strictly after any incident edge, except $\{a_{i,\ell}t_i: 1 \leq \ell \leq 3\}$ or $\{b_{i,\ell}f_i: 1 \leq \ell \leq 3\}$.
\end{claim}
\begin{proof}[Proof of Claim \ref{claim:truth_edges}]
Consider the reachability set of $w_i$.  This certainly contains $w_i$ itself and the $21$ leaf neighbours of $w_i$ as well as $t_i$, $f_i$ and (by Claims \ref{claim:a_edges} and \ref{claim:b_edges}) $a_{i,1},a_{i,2},a_{i,3}$ and $b_{i,1},b_{i,2},b_{i,3}$, a total of $30$ vertices.  It follows that $w_i$ cannot reach any vertices other than these, and hence the remaining edges incident with $t_i$ must be active strictly before $w_it_i$, and similarly for $f_i$.
\renewcommand{\qedsymbol}{$\square$ (Claim \ref{claim:truth_edges})}
\end{proof}

We now return our attention to the unsatisfied clause $C_j$, and consider the vertex corresponding to an arbitrary literal $\ell_r$ in the clause; suppose that $\ell_r \in \{x_i, \neg x_i\}$, and let $u_r$ be the vertex corresponding to $\ell_r$ (either $t_i$ or $f_i$).  Observe that $u_r$ certainly reaches itself and all of its $25$ neighbours.  Moreover, by definition of $b$, since $\ell_r$ is false, it follows that $u_r$ also reaches the unique vertex $u_r'$ in $\{t_i,f_i\} \setminus u_r$ and hence, by the observations above, also reaches the three neighbours of $u_r'$ in the set $\{a_{i,1},a_{i,2},a_{i,3},b_{i,1},b_{i,2},b_{i,3}\}$.  The reachability set of $u_r$ therefore has cardinality at least $30$; to avoid exceeding this, we conclude that $u_r$ does not reach $v_j$, implying that the edge incident with $v_j$ on the two-edge path from $v_j$ to $u_r$ is the first active edge on this path.  

It therefore follows that $v_j$ reaches $u_r$.  Moreover, by Claims \ref{claim:a_edges}, \ref{claim:b_edges} and \ref{claim:truth_edges}, we know that $v_j$ reaches $u_r$ strictly before any of the edges in the set $\{t_iw_i, f_iw_i, t_ia_{i,1}, t_ia_{i,2}, t_ia_{1,3}, f_ib_{i,1}, f_ib_{i,2}, f_ib_{i,3}\}$ is active.  Since $u_r$ reaches every vertex in the set $\{t_i,w_i,f_i,a_{i,1},a_{i,2},a_{i,3},b_{i,1},b_{i,2},b_{i,3}\}$ using only this set of edges, it follows that $v_j$ also reaches these vertices.  Thus we see that $v_j$ reaches at least $9$ vertices via temporal paths that include $u_r$. Since this holds for every literal $\ell_r$ in $C_j$, and $v_j$ also reaches itself and its three neighbours, we conclude that $v_j$ has a reachability set of cardinality at least
$$1 + 3 + (3 \times 9) = 31,$$
giving the required contradiction to the assumption that no vertex has a reachability set of cardinality greater than $30$ under the assignment $t$.
\end{proof}

\begin{proof}[Proof of Theorem \ref{thm:sing-hard-new}]
The result follows immediately from Lemmas \ref{lma:truth->order} and \ref{lma:order->truth}.
\end{proof}

The reduction used to prove Theorem \ref{thm:sing-hard-new} immediately implies a lower bound on the best approximation ratio we can hope to achieve in polynomial time for the optimisation version of the problem, even when the input graph has bounded degree.  Note that, in the reduction, we can achieve maximum temporal reachability at most $30$ if the $(3,2B)$-\textsc{SAT} formula is satisfiable, and otherwise the maximum reachability with respect to any edge ordering must be at least $31$.  Therefore, any algorithm that allows us to compute a $c$-approximation to the minimum possible cardinality of the largest reachability set, where $c < 31/30$, would determine whether the $(3,2B)$-\textsc{SAT} formula is satisfiable (it is satisfiable if and only if the output of the approximation algorithm is strictly less than $31$).

\begin{cor}\label{cor:approx-hard}
Unless $\P=\NP$, there is no polynomial-time approximation algorithm for the optimisation version of \singTempOrd\, with performance ratio better than $31/30$.
\end{cor}

\subsection{General bounds and an approximation algorithm}\label{sec:singleton-approx}

In this section we show that, in the singleton case, it is always possible to find an ordering of the edges so that the maximum reachability is bounded by a function of the edge chromatic number of the input graph, and that this bound is in fact tight for certain graphs.  As a consequence, we obtain a constant-factor approximation algorithm for bounded degree graphs.  We begin with a general result showing that we can bound the maximum reachability by considering the maximum reachability of subgraphs that partition the edge-set.

\begin{lma}\label{lma:reach-decomp}
Let $G = (V,E)$ be a graph and let $\mathcal{P} = \{E_1,\ldots,E_s\}$ be a partition of $E$.  Let $G_i = (V, E_i)$ be the subgraph with edge set $E_i$ for each $1 \leq i \leq s$, and suppose that for each $i$ there is an assignment $t_i$ of times $\{1,\ldots,|E_i|\}$ to elements of $E_i$ such that the maximum reachability of $(G_i,E_i,t_i)$ is at most $r_i$.  Then there is an assignment $t$ of times $\{1,\ldots,|E|\}$ to the edges in $E$ such that the maximum reachability of $(G,E,t)$ is at most $\prod_{i=1}^s r_i$.
\end{lma}
\begin{proof}
We proceed by induction on $s$.  The base case, for $s=1$, is trivially true, so suppose that $s > 1$ and that the result holds for any $s' < s$.

Let $E' = \bigcup_{i = 1}^{s-1} E_i$, and $G' = (V, E')$.  We can then apply the inductive hypothesis to $G'$, $E'$ and $\mathcal{P}' = \{E_1,\ldots,E_{s-1}\}$ to see that there exists an assignment $t'$ of times to the edges in $E'$ such that the maximum reachability of $(G',E',t')$ is at most $\prod_{i=1}^{s-1} r_i$.

Now consider an assignment $t$ of times to the edges in $E$, such that $t(e_j) = t'(e_j)$ if $e_j \in E'$, and otherwise $t(e_j) = t_s(e_j) + |E'|$.  Note that, for any $v \in V$, the set of vertices that $v$ reaches in $(G,E,t)$ after time $|E'|$ is precisely $\reach_{G_s,E_s,t_s}(v)$.

We claim that the maximum reachability of $(G,E,t)$ is at most $\prod_{i=1}^s r_i$, as required.  To see that this is true, we fix an arbitrary vertex $v \in V$, and argue that $|\reach_{G,E,t}(v)| \leq \prod_{i=1}^s r_i$.  Set $U = \reach_{G',E',t'}(v)$ and recall that $|U| \leq \prod_{i=1}^{s-1}r_i$.  Now let $w \in \reach_{G,E,t}(v) \setminus U$, and observe that there must be a strict temporal path from some $u \in U$ to $w$ that uses only edges of $E_s$; equivalently, $w \in \reach_{G_s,E_s,t_s}(u)$.  We therefore see that
$$|\reach_{G,E,t}(v)| = \left|\bigcup_{u \in U} \reach_{G_s,E_s,t_s}(u) \right| \leq |U| \cdot r_s \leq \prod_{i=1}^s r_i, $$
completing the proof.
\end{proof}

We use this result to obtain an upper bound on the maximum reachability that can be achieved, in terms of the edge chromatic number of $G$.  The edge chromatic number of $G = (V,E)$, written $\chi'(G)$, is the smallest $c \in \mathbb{N}$ such that there is a colouring $f: E \rightarrow [c]$ such that $f(e_1) \neq f(e_2)$ whenever $e_1$ and $e_2$ are incident.

\begin{lma}\label{thm:edge-col}
Given any graph $G = (V,E)$, there is assignment $t: E(G) \rightarrow [|E(G)|]$ of times to edges such that the maximum reachability of $(G,E,t)$ is at most $2^{\chi'(G)}$.
\end{lma}
 \begin{proof}
 Fix a proper edge colouring $c: E(G) \rightarrow \chi'(G)$ of $G$, and suppose that $E_i$ is the set of edges receiving colour $i$ under $c$.  Since $G_i = (V,E_i)$ consists of disjoint edges and perhaps isolated vertices, the largest connected component of $G_i$ contains at most two vertices, and hence for any assignment $t_i$ of times to edges in $E_i$ we have that the maximum reachability of $(G_i,E_i,t_i)$ is at most two.  The result now follows immediately from Lemma \ref{lma:reach-decomp}.
 \end{proof}

We now argue that the bound in Lemma \ref{thm:edge-col} is tight for paths.

\begin{prop}\label{prop:path-lb}
Let $P = (V_P,E_P)$ be a path on at least five vertices.  Then the minimum value of the maximum reachability of $(V_P,E_P,t)$, taken over all possible assignments of times to edges, is equal to four.
\end{prop}
 \begin{proof}
The fact that we can achieve maximum reachability at most four follows immediately from Lemma \ref{thm:edge-col}, since any path admits a proper edge-colouring with two colours.  It therefore remains to demonstrate that no assignment of times to edges can achieve maximum reachability strictly less than four.

To see this, let $e_1 = uv$ and $e_2 = vw$ be two incident, non-leaf edges; we may assume without loss of generality that $t(e_1) < t(e_2)$.  Recall from Observation \ref{prop:degree} that $u$ necessarily reaches all of its neighbours, so $|\reach_{V_P,E_P,t}(u) \setminus \{w\}| \geq 3$; however, the fact that $t(e_1) < t(e_2)$ means that there is a strict temporal path from $u$ to $w$ and so $w \in \reach_{V_P,E_P,t}(u)$.  Hence $|\reach_{V_P,E_P,t}(u)| \geq 4$, as required.
 \end{proof}

We conclude this section by using Lemma \ref{thm:edge-col} (combined with Observation \ref{prop:degree}) to derive a polynomial-time approximation algorithm, whose optimisation ratio depends only on the maximum degree of the input graph.

\begin{thm}\label{thm:deg-approx}
Given any graph $G$, we can compute a $\frac{2^{\Delta(G) + 1}}{\Delta(G) + 1}$-approximation to the optimisation version of \singTempOrd\, in time $\mathcal{O}(|E(G)|)$, where $\Delta(G)$ denotes the maximum degree of $G$.  Moreover, we can also compute an assignment of times to edges which achieves this approximation ratio in $\mathcal{O}(|E(G)||V(G)|)$ time.
\end{thm}
 \begin{proof}
 We claim that it suffices to compute the maximum degree $\Delta(G)$ (which can be done in time linear in the number of edges of $G$) and to return $2^{\Delta(G) + 1}$.  To see that this satisfies the requirements, let $\opt(G)$ denote the smallest value of the maximum reachability of $(G,E,t)$ taken over all assignments $t$ of times to edges.  We know from Lemma \ref{thm:edge-col}, together with the fact that the edge chromatic number of any graph is at most the maximum degree plus one, that $\opt(G) \leq 2^{\Delta(G) + 1}$; conversely by Observation \ref{prop:degree} we know that $\opt(G) \geq \Delta(G) + 1$.  Hence 
$$ \opt(G) \leq 2^{\Delta(G) + 1} \leq \frac{2^{\Delta(G) + 1}}{\Delta(G) + 1} \opt(G),$$
as required. 

To see that we can compute a suitable assignment of times to edges in polynomial time, we observe that a $(\Delta + 1)$-edge colouring of any graph can be constructed in $\mathcal{O}(|E(G)||V(G)|)$ time \cite{jayadev1992}; given such a colouring we follow the method of Lemma \ref{thm:edge-col} to construct a suitable assignment of times in an additional $\mathcal{O}(|E(G)|)$ time, giving an overall $\mathcal{O}(|E(G)||V(G)|)$ time complexity.
\end{proof}

\begin{cor}\label{cor:bdd-deg-approx}
Given any graph $G$ of bounded degree, a constant-factor approximation to the optimisation version of \singTempOrd\, can be computed in time $\mathcal{O}(|E(G)|)$.
\end{cor}

We know from Corollary \ref{cor:approx-hard} that we cannot hope to obtain an arbitrarily good approximation in polynomial time (and in particular there is no $\PTAS$ unless $\P$=$\NP$) even on graphs of bounded degree; however, it is natural to ask whether we can nevertheless make a significant improvement on the approximation algorithm we have described here.

\begin{openproblem}
Is there a polynomial-time algorithm to compute a $c$-approximation to the optimisation version of \singTempOrd, where $c > \frac{31}{30}$ is a constant independent of the maximum degree of the input graph?
\end{openproblem}

\subsection{Exact algorithms for special cases}\label{sec:singleton-special}

In this section, we demonstrate that certain restrictions on the input graph lead to efficient exact algorithms in the singleton edge-class setting. Where possible, we have chosen graph classes with special epidemiological motivation.  For example, directed acyclic graphs, while very simple, have the capacity to capture a number of animal trading systems: beef cattle are typically produced by one type of farm, grown by another, and finished by a third.  The pig-rearing industry is similar \cite{schulz2017network}.  Thus, graphs modelling these systems are unlikely to have many disjoint directed cycles over short periods of time; it would therefore be reasonable to assume that such graphs might admit small directed feedback vertex sets, but in the first instance we consider the simplest case in which the graph is in fact a DAG.  We also consider trees as a first step towards understanding the complexity of the problem on graphs of limited treewidth: trees can locally approximate most sparse contact systems, and we have previously found that at least some livestock trading systems can be modelled by graphs of limited treewidth \cite{edge-deletion}.

\subsubsection{Directed acyclic graphs}
We begin by giving a simple necessary and sufficient condition for a yes-instance when the input graph is a DAG, implying that the problem is solvable in polynomial time in this case.

\begin{thm}\label{thm:dag-min-trivial}
Let $(G,\mathcal{E},k)$ be an instance of \singTempOrd\, where $G = (V,\overrightarrow{E})$ is a DAG.  Then $(G,\mathcal{E},k)$ is a yes-instance if and only if $k \geq \Delta^{\out} + 1$, where $\Delta^{\out}$ is the maximum out-degree of $G$.
\end{thm}
\begin{proof}
We know by Observation \ref{prop:out-degree} that, for any bijection $t \colon \mathcal{E} \rightarrow [|\mathcal{E}|]$, the maximum reachability of $(G,\mathcal{E},t)$ will be at least $\Delta^{\out} + 1$, so if $k \leq \Delta^{\out}$ we will certainly have a no-instance.

Conversely, we argue that there is an ordering $t \colon \mathcal{E} \rightarrow [|\mathcal{E}|]$ such that the maximum reachability of $(G,\mathcal{E},t)$ is at most $\Delta^{\out} + 1$.  Fix a topological ordering $\{v_1,\ldots,v_m\}$ of the vertices of $G$.  We now define a related ordering of the edges: fix any ordering $(e_1,\ldots,e_m)$ of $E(G)$ such that $\overrightarrow{v_iu}$ precedes $\overrightarrow{v_jw}$ whenever $v_i$ precedes $v_j$ in the topological ordering.  If $E_i = \{e_i\}$ for each $i$, we now define $t(e_i) = m + 1 - i \in [m]$.  

We claim that there is no strict temporal path on more than one edge in $(G,\mathcal{E},t)$.  Suppose, for a contradiction that $u,v,w$ is a strict temporal path in $(G,\mathcal{E},t)$.  In this case we must have $t(\overrightarrow{vw}) > t(\overrightarrow{uv})$, which by definition of $t$ means that $\overrightarrow{vw}$ precedes $\overrightarrow{uv}$ in our edge ordering; this implies that $v$ precedes $u$ in the topological ordering of the vertices, but as $v$ is reachable from $u$ this is not possible.  Thus we can conclude that the longest temporal path in $(G,\mathcal{E},t)$ consists of a single edge, so the reachability set of any vertex $v$ is contained in $\{v\} \cup N_G^{out}(v)$, as required.
\end{proof}

We note that the argument used to prove Theorem \ref{thm:dag-min-trivial} cannot easily be extended to graphs with small directed feedback vertex sets: the degree of the vertices we delete to obtain a DAG will be crucial.  We leave the investigation of the this problem parameterised by feedback vertex (or edge) set size as future work.

\subsubsection{Trees}

Recall from Theorem \ref{thm:sing-hard-new} that the problem remains NP-complete on general graphs even when the maximum permitted reachability is bounded by a constant.  We now consider this setting with the additional restriction that the input graph is a tree, and show that we can solve \singTempOrd\, in polynomial time; in fact we give an FPT algorithm with respect to the parameter $k$ (maximum permitted reachability set size).  We use a dynamic programming approach, working from the leaves to an arbitrarily chosen root vertex of the tree. For each vertex $v$, we record a set of states that captures concisely the relevant information about all possible orderings of the edges in the subtree rooted at~$v$.

\begin{thm}\label{thm:trees-poly}
When the underlying graph $G$ is a tree, \singTempOrd\, can be solved in time $n \cdot k^{\mathcal{O}(k)}$.
\end{thm}
\begin{proof}
Let $(T=(V_T,E_T),E_T,k)\}$ be the input to our instance of \singTempOrd.  We fix an arbitrary root vertex $v_r \in V_T$, and for each $v \in V_T$ we denote by $T_v$ the subtree of $T$ rooted at $v$.  For every vertex $v \in V_T$ other than the root, we refer to the first edge on the path from $v$ to $v_r$ as $e_{v \uparrow}$.  By Observation \ref{prop:degree}, we may assume that the maximum degree of $T$ is at most $k-1$.

For each vertex $v$ in $V_T$, we record a set of states; a state is a pair $(\alpha, \beta)$ where $\alpha, \beta \in [k]$.  For vertex $v$, we say a state $(\alpha, \beta)$ is \emph{realisable} if there is an ordering of the edges incident at vertices in $T_v$ such that:
\begin{enumerate}
\item $e_{v \uparrow}$ is assigned time $\tau_{\uparrow}$,
\item the largest reachability set within $T_v$ that reaches $v$ before $\tau_{\uparrow}$ is of size~$\alpha$,
\item the number of vertices (including $v$) reachable within $T_v$ from $v$ after $\tau_{\uparrow}$ is exactly $\beta$, and 
\item there is no reachability set within $T_v$ under this temporal ordering of size more than $k$.  
\end{enumerate}
 Observe that there are at most $k^2$ possible states for any vertex.

It is straightforward to generate all realisable states at a leaf $u$: the realisable states here are precisely the states $(1,1)$.  Now consider a non-leaf vertex $v$.  Assume that we have a list of all realisable states for each child of $v$.  We argue that we can efficiently find all realisable states of $v$ from the realisable states of its children.  

We will reason about a joint state $\{(\alpha_1, \beta_1),\ldots,(\alpha_d, \beta_d)\}$ of the children $v_1, \ldots, v_d$ of $v$, along with an ordering $\Pi$ of the edges incident at $v$.   
A state $(\alpha, \beta)$ of a non-leaf non-root vertex $v$ is realisable if and only if there is a joint state $\{(\alpha_1, \beta_1),\dots,(\alpha_d, \beta_d)\}$ of the children of $v$ and an ordering $\Pi$ of the edges incident at $v$ such that: 
\begin{itemize}
\item $\alpha = \max_{i, \Pi(vv_i)< \Pi(e_{v\uparrow})} \left\lbrace \sum_{j, \Pi(vv_j)> \Pi(vv_i)}  \beta_j + \alpha_i + 1 \right\rbrace$
\item $\beta = 1 + \sum_{i, \Pi(vv_i)> \tau_{\uparrow}} \beta_i$ ,
\item for all $1 \leq i \leq d$,  $1 + \alpha_i + \sum_{j, \Pi(vv_j)>  \Pi(vv_i)} \beta_j \leq k$, and
\item $1 + \sum_i \beta_i \leq k$.
\end{itemize}
At the root vertex $v_r$ we follow a similar procedure, but without reference to the edge to a parent: a given state $(\alpha,\beta)$ of $v_r$ is realisable if and only if there exists some joint state $\{(\alpha_1, \beta_1),\ldots,(\alpha_d, \beta_d)\}$ of the children of the root, along with an ordering $\Pi$ of the edges to those children such that:
\begin{itemize}
\item $\alpha = \max_{i} \left\lbrace \sum_{j, \Pi(v_rv_j)> \Pi(v_rv_i)} \beta_j  + \alpha_i + 1 \right\rbrace$
\item $\beta = 1 + \sum_i \beta_i $, and
\item for all $1 \leq i \leq d$,  $1 + \alpha_i + \sum_{j, \Pi(v_rv_j) > \Pi(v_rv_i)}  \beta_j \leq k$.
\end{itemize}
Thus, in order to compute the list of realisable states for a vertex $v$ with $d$ children in the rooted tree, we consider all possible joint child states and determine and store all possible parent states that are consistent with at least one joint child state.  A joint child state consists of an ordering of the edges incident at a vertex, for which there are at most $d!$ possibilities, together with a choice of one of the realisable states for each child (where each child has at most $k^2$ realisable states).  Therefore the number of possible joint child states and ordering combinations to consider is at most $d!k^{2d}$.  For each joint child state and ordering, we can determine the corresponding states of the parent in time $O(d^2)$.  Thus the total time required at each vertex is $\mathcal{O}(d!k^{2d}d^2)$.  By our initial assumption that the maximum degree is at most $k-1$, it follows that the time required at any vertex is $\mathcal{O}(k!k^{2k}k^2) = k^{\mathcal{O}(k)}$.  If any realisable state exists at the root of the tree, there is an acceptable ordering of edges in the tree.  It follows that we can solve \singTempOrd\, in time $n \cdot k^{\mathcal{O}(k)}$.
\end{proof}

Note that the procedure described can easily be adapted to output a suitable ordering of the edges, if one exists: at each vertex $v$, we concatenate the orderings of edges in the subtrees rooted at its children in such a way that the relative ordering $\Pi$ of edges incident with $v$ is preserved.

We conjecture that this approach can be generalised to give an algorithm for \singTempOrd~that is FPT parameterised by treewidth and maximum reachability set size.  

\begin{conjecture}
There is an algorithm for \singTempOrd~that is FPT by treewidth and maximum permissible reachability set size.
\end{conjecture}

To support further work on this question, we suggest a possible state that could be used in a standard dynamic programming approach on a nice tree decomposition (for examples of these approaches and the definition of a nice tree decomposition, we refer the reader to \cite{paramalgs}).

We suggest adapting the state used in the proof of Theorem \ref{thm:trees-poly}, but this becomes much more complicated when dealing with tree decompositions.  Let $B$ be a bag in a nice tree decomposition, and $E_B$ the set of all edges incident at members of $B$.  Let $\mathcal{P}(B)$ be the powerset of $B$, and $\mathcal{D}$ be the set of all functions with domains that are members of $\mathcal{P}(B)$ and ranges that are subsets of $[|E_B|]$.  Then a state for bag $B$ could consist of:
\begin{itemize}
  \item $\Pi$: an assignment of times from $\{1, \ldots, |E_B|\}$ to the members of $E_B$, essentially an ordering of these edges, and
  \item $\mathcal{V}^T$: a function from the cross product $\mathcal{D} \times \mathcal{D}$ to a natural number between 1 and $k$.
\end{itemize}
A state $(\Pi, \mathcal{V}^T )$ for bag $B$ would be \emph{realisable} if there exists an ordering of the edges incident at vertices present in bags of the subtree rooted at $B$ such that:
\begin{itemize}
\item the ordering is consistent with $\Pi$, and
\item under that edge ordering, for every $(U^T, W^T) \in \mathcal{D} \times \mathcal{D}$, we have that 
$\mathcal{V}^T((U^T, W^T))$ is the size of the union of 
\begin{itemize}
\item the largest reachability set of a vertex in $B$ or one of its descendants that reaches exactly those vertices in the domain of $U^T$, by their times assigned by $U^T$, and no other vertices in $B$, and 
\item the size of the joint reachability set of vertices in the domain of $W^T$ after their times assigned by $W^T$. 
\end{itemize}
\end{itemize}
 The need for the complexity of this state comes from one complicating case: essentially the vertex added at an introduce node may be reached by some existing large reachability set and itself also reach a large reachability set, both of which contain vertices that are present only in the descendants of the node, but not at the introduce node's bag itself.  To perform the necessary accounting, we must know the size of the union of these two sets: they may have many vertices in common.  This motivates the requirement for the function $\mathcal{V}^T$, and the consideration of all pairs of subsets of the bags, with times assigned to their members.   We hope that a more elegant approach may be found in future work. 

We leave as an open problem the more general question on trees when no restriction is placed on the maximum permitted reachability set size.
\begin{openproblem}
Is \singTempOrd\ ~$\NP$-complete when restricted to trees?
\end{openproblem}

Although it remains open whether the problem can be solved exactly in polynomial time on trees of unbounded degree, we are able to improve on our approximation algorithm in this special case.  To demonstrate this, we begin by demonstrating that we can always find an ordering of the edges of a tree such that every reachability set is contained in the second neighbourhood of some vertex.  For any vertex $v$, let $N_2[v]$ denote the closed second neighbourhood of $v$ (i.e.~$v$ itself and all vertices at distance at most two from $v$).

\begin{lma}\label{lma:alternating-layers}
Let $T = (V_T, E_T)$ be a tree.  For any vertex $v \in V_T$, denote by $U_v \subset V_T$ the vertices that are at even distance from $v$ in $T$.  There is an assignment $t$ of times to edges in $E_T$ such that, for every vertex $u \in V_T$, there exists $u_v \in U_v$ such that $\reach_{T,E_T,t}(u) \subseteq N_2[u_v]$.
\end{lma}
\begin{proof}
Root the tree $T$ at vertex $v$.  We begin by partitioning the edges of $T$ into two sets: for every edge $e \in E_T$, we call $e$ \emph{even} if the endpoint of $e$ that is closest to $v$ belongs to $U_v$, and \emph{odd} otherwise.  Note that the path between any vertex and one of its descendants in the tree will consist of alternate even and odd edges.  We now define $\mathcal{T}$ based on this partition: fix $t$ to be any assignment of times such that all even edges are assigned times before any odd edge.  Note then that no temporal path can include any odd edge followed by an even edge.  We now argue that the temporal graph $(T,E_T,t)$ has the required properties.

We consider two cases: first, let $u \in U_v$.  In this case, the edge $up_u$ between $u$ and its parent $p_u$ is odd, and thus is assigned a time after every even edge, including the edge from $p_u$ to its parent. Hence the only vertices that $u$ might reach via the edge $up_u$ are $p_u$ and its children, all within the second-neighbourhood of $u$ (note that all other edges incident with children of $p_u$ are even and so are assigned earlier times than the edge from such a child to $p_u$).  We turn our attention to the descendants of $u$: because of the odd/even alternation structure, $u$ can only reach descendants at most distance two from $u$, which are also within its second neighbourhood.  Thus $\reach_{T,E_T,t}(u) \subseteq N_2[u_v]$, as required.

Now, consider a vertex $w \notin U_v$: the edge $wp_w$ between $w$ and its parent is even, and the edges between $w$ and its children are odd.  We will argue that in this case $\reach_{T,E_T,t}(w) \subseteq N_2[p_w]$.  Since the edges from the children of $w$ to its grandchildren are even and are therefore assigned times before the edges from $w$ to its children, the vertex $w$ reaches its children but no other descendants, and so all descendants of $w$ in $\reach_{T,E_T,t}(w)$ also belong to $N_2[p_w]$.  It remains to show that all vertices reachable from $w$ that are not descendants of $w$ also belong to $\reach_{T,E_T,t}$.  Such vertices must be reached via a temporal path that includes $p_w$ and so in particular also belong to $\reach_{T,E_T,t}(p_w)$.  However, we know that $p_w \in U_v$ and so, from the previous case, $\reach_{T,E_T,t}(p_w) \subseteq N_2[p_w]$.  We therefore conclude that $\reach_{T,E_T,t}(w) \subseteq N_2[p_w]$, as required.
\end{proof}

This result immediately gives us an upper bound on the smallest achievable maximum reachability for trees, since the closed second neighbourhood of a vertex $v$ is at most $\Delta^2 +1$, where $\Delta$ denotes the maximum degree of $T$.

\begin{cor}\label{cor:trees-degbound}
Given any tree $T$, there is an assignment $t$ of times to edges of $T$ such that the maximum reachability of $(T,E(T),t)$ is at most $\Delta(T)^2+1$, where $\Delta(T)$ denotes the maximum degree of $T$.
\end{cor}

Moreover, the proof of Lemma \ref{lma:alternating-layers} gives a method for constructing, in linear time, an ordering of the edges under which every reachability set is contained in the closed second neighbourhood of some vertex; combined with the fact (Observation \ref{prop:degree}) that the maximum reachability of a graph with maximum degree $\Delta$ is necessarily at least $\Delta + 1$, this gives rise to a linear-time approximation algorithm for trees.

\begin{cor}\label{cor:trees-approx}
Given any tree $T$, we can compute a $\Delta$-approximation to the optimisation version of \singTempOrd\, in linear time, where $\Delta$ denotes the maximum degree of~$\,T$.  Moreover, we can also compute an assignment of times to edges which achieves this approximation ratio in linear time.
\end{cor}

While this approximation algorithm achieves a much better ratio on general trees than that of Theorem \ref{thm:deg-approx}, we note that Lemma \ref{thm:edge-col} in fact gives a better bound than Corollary \ref{cor:trees-degbound} for trees of maximum degree exactly three (since the edge chromatic number of any tree is equal to the maximum degree, so Lemma \ref{thm:edge-col} gives a bound of $2^\Delta = 8$, compared with a bound of $\Delta^2 + 1 = 10$ from Corollary \ref{cor:trees-degbound}).  In fact, we now show that the bound of Lemma \ref{thm:edge-col} is tight for sufficiently large binary trees.

\begin{lma}\label{prop:binary-lb}
Let $B = (V_B,E_B)$ be a rooted binary tree of depth at least $14$, and let $t$ be any assignment of times to edges of $B$.  Then the maximum reachability of $(B,E_B,t)$ is at least $8$.
\end{lma}
\begin{proof}
We will assume, for a contradiction, that no vertex has a reachability set of size greater than seven in $(B,E_B,t)$.

We claim that $B$ must contain some edge $e_1$ such that:
\begin{enumerate}
\item $t(e_1) < t(e')$ for every edge $e'$ incident with $e_1$,
\item both endpoints of $e_1$ are at distance at least four from any leaves, and
\item both endpoints of $e_1$ are at distance at least four from the root.
\end{enumerate}
To find such an edge, we start at an arbitrary vertex $s$ at distance three from the root (i.e.~at depth four); we will construct a path starting at $s$ which leads away from the root.  Each time we reach a new vertex, we choose our next edge to be the edge (out of the two possibilities) which is assigned the earlier time.  We stop when we reach a vertex $t$ at distance three from a leaf, and call the resulting path $P$.  

We say that an edge $e$ on $P$, not containing $s$ or $t$, is a \emph{minimum edge} if both other edges of $P$ incident with $e$ are assigned times strictly later than that assigned to $e$.  If $P$ does not contain a minimum edge, then it consists of a (possibly empty) segment on which the assigned times of edges increase along the path, followed by a (possibly empty) segment on which the assigned times of edges decrease along the path.  If the first segment contains more than three edges, we have a reachability set of size at least nine (since the first vertex on the path must reach both children of the first four vertices on $P$, including itself).  Similarly, if the second segment contains more than three edges, the last vertex on $P$ will reach the previous four vertices on the path together with each such vertex's child outside $P$, a total of nine vertices.  Therefore we may assume that both segments contain at most three edges.  However, since the tree has depth at least $14$, we can conclude that $P$ contains at least $7$ edges, a contradiction.  We may therefore assume that $P$ contains a minimum edge, $e_1$.

We now argue that $e_1$ has the required properties.  It is clear from the construction of $P$ and the fact that $e_1$ does not contain either endpoint of $P$ that conditions (2) and (3) are satisfied.  Condition (1) follows from the definition of a minimum edge together with the construction of $P$: since $e_1$ is a minimum edge, it must be assigned a time before the earlier of the two incident edges below $e_1$ in the rooted binary tree (as the earlier of these is incident with $e_1$ in $P$); the definition of a minimum edge means that $e_1$ is also assigned a time earlier than the edge above it (the other edge incident with $e_1$ in $P$); if the remaining edge incident with $e_1$ was assigned a time earlier than $e_1$ we would have chosen $P$ to include this edge instead.

We continue the argument using this choice of edge $e_1$; we will reason about a subtree including $e_1$, which is illustrated in Figure \ref{fig:binary}.  Let $v$ and $v'$ be the endpoints of $e_1$, and note that both $v$ and $v'$ must have the same reachability set (applying Observation \ref{prop:first-edge} twice).  It is clear that this reachability set contains $v$, $v'$ and all neighbours of either $v$ or $v'$ (see highlighted vertices in Figure \ref{fig:binary}).  Thus, if the reachability set of $v$ has size at most seven, it can contain at most one more vertex; we may therefore assume without loss of generality that any other vertex in the reachability set of $v$ is at lesser distance from $v'$ than $v$ (otherwise, by symmetry, we may swap the roles of $v$ and $v'$).

\begin{figure}
\centering
\includegraphics[width = 0.3 \linewidth]{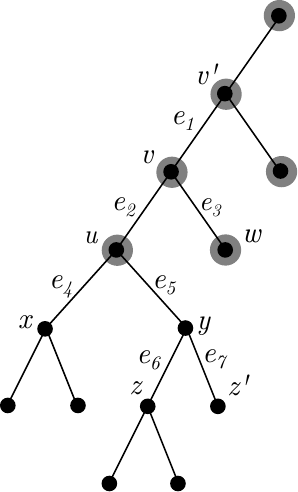}
\caption{A subtree of a binary tree, used to argue that the largest reachability set must have size at least eight.  If $e_1$ is a minimum edge, then all highlighted vertices must belong to the reachability sets of both $v$ and $v'$.}
\label{fig:binary}
\end{figure}

Let $e_2 = vu$ and $e_3 = vw$ be the other two edges containing $v$, where $t(e_2) < t(e_3)$.  Further let $e_4 = ux$ and $e_5 = uy$ be the other two edges containing $u$, with $t(e_4) < t(e_5)$.  By our previous assumptions, $v$ does not reach $x$ or $y$, so we must have that $t(e_4),t(e_5) < t(e_2)$.  Finally, let $e_6 = yz$ and $e_7 = yz'$ be the other two edges containing $y$, where $t(e_6) < t(e_7)$.

Suppose first that $t(e_6) < t(e_5)$; we will argue that in this case the reachability set of $z$ will have size at least eight.  To see this, note that $z$ reaches itself and all its neighbours as well as $z'$, $u$, $v$ and $w$.  Thus we may assume from now on that $t(e_6) > t(e_5)$ (and hence also $t(e_7) > t(e_5)$).  Now consider the reachability set of $x$.  We see that $x$ reaches itself and all its neighbours, as well as $y$, $z$, $z'$, $v$ and $w$, giving $x$ a reachability set of size at least nine, which is a contradiction.  Therefore we conclude that the maximum reachability of $(B,E_B,t)$ is at least $8$, as required.
\end{proof}

\section{The general problem}
\label{sec:general}

In this section we see that the tractable cases we identified in Section \ref{sec:singleton} do not extend to the more general setting where edge classes of cardinality greater than one are allowed.  We begin by complementing Theorem \ref{thm:dag-min-trivial} with a hardness result for DAGs.

\begin{thm}\label{thm:no-list-hardness}
\tempOrd\, is $\NP$-complete, even if $G$ is a DAG whose underlying graph has maximum degree at most 5, $k$ is at most 9, and $|E_i| \leq 3$ for each $E_i \in \mathcal{E}$.
\end{thm}
\begin{proof}
We provide a reduction from the following problem, shown to be $\NP$-complete in \cite{tovey84}.

\begin{framed}
\noindent
\textbf{\textsc{(3,4)-SAT}}\\
\textit{Input:} A CNF formula $\Phi$ in which every clause contains exactly three distinct variables, and every variable appears in at most four clauses.\\
\textit{Question:} Is $\Phi$ satisfiable?
\end{framed}

Let $\Phi = C_1 \wedge \cdots \wedge C_m$ be our instance of \textsc{(3,4)-SAT}, and suppose that the variables in $\Phi$ are $x_1,\ldots,x_n$.  We construct an instance $(G,\mathcal{E},k)$ (with the properties in the statement of the theorem) which is a yes-instance if and only if $\Phi$ is satisfiable.

The vertex set of $G$ consists of two sets, $V_{clause} = \{c_j: 1 \leq j \leq m\}$, and $V_{var} = \{v_{x_i,1},v_{x_i,2},v_{x_i,3},v_{\neg x_i,1},v_{\neg x_i,2},v_{\neg x_i,3}: 1 \leq i \leq n\}$.  $G$ contains directed edges $\overrightarrow{v_{x_i,1}v_{x_i,2}}$, $\overrightarrow{v_{x_i,2}v_{x_i,3}}$, $\overrightarrow{v_{\neg x_i,1}v_{\neg x_i,2}}$ and $\overrightarrow{v_{\neg x_i,2}v_{\neg x_i,3}}$ for each $1 \leq i \leq n$; for each $1 \leq j \leq m$, if $C_j = (\ell_1 \vee \ell_2 \vee \ell_3)$, we also have edges $\overrightarrow{c_jv_{\ell_1,1}}$, $\overrightarrow{c_jv_{\ell_2,1}}$ and $\overrightarrow{c_jv_{\ell_3,1}}$.

We now define the set $\mathcal{E}$ of edge-classes.  For each clause $C_j$ and literal $\ell$ appearing in $C_j$, we have four sets in $\mathcal{E}$:
\begin{itemize}
\item two copies of the set $\{\overrightarrow{c_j,v_{\ell,1}},\overrightarrow{v_{\ell,2}v_{\ell,3}},\overrightarrow{v_{\neg \ell,1}v_{\neg \ell,2}}\}$, denoted $E_{C_j,\ell}^{(1)}$ and $E_{C_j,\ell}^{(2)}$, and
\item two copies of the set $\{\overrightarrow{c_j,v_{\ell,1}},\overrightarrow{v_{\ell,1}v_{\ell,2}},\overrightarrow{v_{\neg \ell,2}v_{\neg \ell,3}}\}$, denoted $E_{C_j,\neg \ell}^{(1)}$ and $E_{C_j,\neg \ell}^{(2)}$.
\end{itemize}

We complete the construction of our instance of \tempOrd\, by setting $k=9$.  Part of the construction is illustrated in Figure \ref{fig:dag-construction}.  It is straightforward to verify that $G$ is a DAG whose underlying graph has maximum degree at most 5.

\begin{figure}
\centering
\includegraphics[width = 0.6 \linewidth]{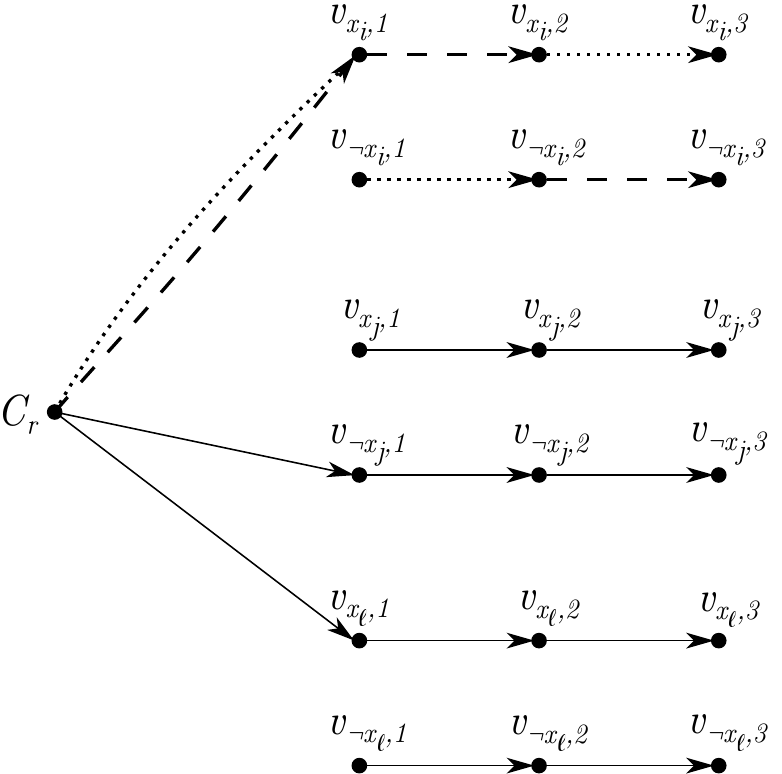}
\label{fig:dag-construction}
\caption{The part of $G$ corresponding to a single clause $C_r = (x_i \vee \neg x_j \vee x_{\ell})$.  Dotted edges belong to $E_{C_r,x_i}^{(1)}$ and $E_{C_r,x_i}^{(2)}$, while dashed edges belong to $E_{C_r,\neg x_i}^{(1)}$ and $E_{C_r,\neg x_i}^{(2)}$.}
\end{figure}

Note that the only vertices with reachability sets of cardinality greater than $3$ in the static directed graph $G$ are those corresponding to clauses, so it suffices to argue that there is a function $t \colon \mathcal{E} \rightarrow [4m]$ such that $|\reach_{G,\mathcal{E},t}(c_j)| \leq 9$ for all $1 \leq j \leq m$ if and only if $\Phi$ is satisfiable.  If $C_j = (\ell_1 \vee \ell_2 \vee \ell_3)$, the reachability set of $c_j$ in the static directed graph $G$ is precisely
$$\{c_j, v_{\ell_1,1}, v_{\ell_1,2}, v_{\ell_1,3}, v_{\ell_2,1}, v_{\ell_2,2}, v_{\ell_2,3}, v_{\ell_3,1}, v_{\ell_3,2}, v_{\ell_3,3}\},$$
which has cardinality $10$.  This gives rise to the following observation.
\begin{obs*}\label{obs:static-subset}
We have $|\reach_{G,\mathcal{E},t}(c_j)| \leq 9$ if and only if the temporal reachability set of $c_j$ in $(G,\mathcal{E},t)$ is a strict subset of its reachability set in $G$.
 \end{obs*}
 
Suppose now that $B \colon \{x_1,\ldots,x_n\} \rightarrow \{\true,\false\}$ is a satisfying assignment for $\Phi$.   Let $t$ be any bijection $\mathcal{E} \rightarrow [4m]$ such that $t(E_{C_j,\ell}^{(1)}), t(E_{C_j,\ell}^{(2)}) \leq 2m$ whenever $B(\ell)$ evaluates to $\true$, and $t(E_{C_j,\ell}^{(1)}), t(E_{C_j,\ell}^{(2)}) \geq 2m+1$ whenever $B(\ell)$ evaluates to $\false$.  Fix an arbitrary clause $C_j$.  Since $B$ is a satisfying assignment for $\Phi$, we know that there is some literal $\ell$ appearing in $C_j$ which evaluates to $\true$ under $B$.  We claim that $v_{\ell,3} \notin \reach_{G,\mathcal{E},t}(c_j)$; by our observation above, this will be sufficient to complete the proof that the maximum reachability is at most $9$.  

To see that this is true, observe that the edge $\overrightarrow{v_{\ell,2}v_{\ell,3}}$ appears only in the sets $E_{C_r,\ell}^{(1)}$ and $E_{C_r,\ell}^{(2)}$ where $C_r$ is a clause that contains the literal $\ell$; since we are assuming that $B(\ell)$ evaluates to $\true$, it follows from the definition that each such set is active only during the first $2m$ timesteps.  On the other hand, $\overrightarrow{v_{\ell,1}v_{\ell,2}}$ appears only in the sets $E_{C_s,\neg \ell}^{(1)}$ and $E_{C_s,\neg \ell}^{(2)}$ where $C_s$ is a clause that contains the literal $\neg \ell$, and so is only active at timesteps greater than or equal to $2m + 1$.  Since the only directed path from $c_j$ to $v_{\ell,3}$ uses the edges $\overrightarrow{v_{\ell,1}v_{\ell,2}}$ and $\overrightarrow{v_{\ell,2}v_{\ell,3}}$ in this order, we see that there cannot be a strict temporal path from $c_j$ to $v_{\ell,3}$ in $(G,\mathcal{E},t)$.  Hence $|\reach_{G,\mathcal{E},t}(c_j)| \leq 9$, as required.

Conversely, suppose that there is a bijection $t \colon \mathcal{E} \rightarrow [4m]$ such that the maximum reachability of $(G,\mathcal{E},t)$ is at most $9$.  We define $\maxtime_t(x_i)$ to be the latest timestep assigned by $t$ to any edge-set of the form $E_{C_j,\ell}^{(r)}$ where $r \in \{1,2\}$ and $\ell \in \{x_i, \neg x_i\}$.  We now define a truth assignment as follows:
\begin{equation*}
B(x_i) = \begin{cases}
				\true 		&\text{if $t^{-1}(\maxtime_t(x_i))$ is of the form $E_{C_j,\neg x_i}^{(r)}$}, \\
				\false		&\text{if $t^{-1}(\maxtime_t(x_i))$ is of the form $E_{C_j,x_i}^{(r)}$.}
		 \end{cases}
\end{equation*}
Now fix an arbitrary clause $C_j$ and suppose that the literal $\ell \in \{x_i,\neg x_i\}$ appears in $C_j$.  We claim that, if $B(\ell)$ evaluates to $\false$, we have $v_{\ell,1},v_{\ell,2},v_{\ell,3} \in \reach_{G,\mathcal{E},t}(c_j)$.  By construction of $B$, we know that $t^{-1}(\maxtime_t(x_i))$ is of the form $E_{C_r,\ell}^{(r)}$ (for some clause $C_r$ which involves the literal $\ell \in \{x_i,\neg x_i\}$) and so includes the edge $\overrightarrow{v_{\ell,2}v_{\ell,3}}$.  By definition of $\maxtime_t(x_i)$, this means there exist distinct timesteps $s_1,s_2 < \maxtime_t(x_i)$ such that $t(E_{C_j,\neg \ell}^{(p)}) = s_p$ for $p \in \{1,2\}$; without loss of generality we may assume that $s_1 < s_2$.  Since $\overrightarrow{c_j,v_{\ell,1}} \in E_{C_j,\neg \ell}^{(1)} = t^{-1}(s_1)$ and $\overrightarrow{v_{\ell,1}v_{\ell,2}} \in E_{C_j,\neg \ell}^{(2)} = t^{-1}(s_2)$, we have a strict temporal path $c_j,v_{\ell,1},v_{\ell,2},v_{\ell,3}$, so we do indeed have $v_{\ell,1},v_{\ell,2},v_{\ell,3} \in \reach_{G,\mathcal{E},t}(c_j)$.  Hence, if every literal in $C_j$ evaluates to $\false$ under $B$, we would have $|\reach_{G,\mathcal{E},t}(c_j)| = 10$, contradicting our assumption that the maximum reachability of $(G,\mathcal{E},t)$ is at most $9$.  Thus we can conclude that every clause contains at least one literal which evaluates to $\true$ under $B$, and so $B$ is a satisfying assignment for $\Phi$.
\end{proof}

Next we show that the general version of the problem is $\W[1]$-hard when parameterised by the vertex cover number of the input graph, even when the input graph is a tree.  While the complexity of the singleton version remains open on trees of unbounded degree, we note that the \singTempOrd\, is trivially in FPT on trees when parameterised by the vertex cover number: it is straightforward to bound the number of non-leaf vertices in any tree by a function of the vertex cover number and so, applying Lemma \ref{lma:leaves-first}, we see that the number of distinct edge orders we need to consider depends only on the vertex cover number.

\begin{thm}\label{thm:trees-hard}
\tempOrd\, is $\W[1]$-hard when parameterised by the vertex cover number of $G$, even if we require $G$ to be a tree.
\end{thm}
\begin{proof}
We prove this result by means of a parameterised reduction from the following problem, shown to be $\W[1]$-hard in \cite{downey95}.

\begin{framed}
\noindent
\textbf{$p$-\textsc{Clique}}\\
\textit{Input:} A graph $G = (V, E)$, and a positive integer $k$.\\
\textit{Parameter:} $k$\\
\textit{Question:} Does $G$ contain a clique on $k$ vertices as a subgraph?
\end{framed}

Let $(G,k)$ be the input to an instance of $\mathbf{p}$-\textsc{Clique}, and suppose that $V(G) = \{v_1,\ldots,v_n\}$ and $E(G) = \{e_1,\ldots,e_m\}$.  We will construct an instance $(G',\{E_1,\ldots,E_h\},k')$ of \tempOrd, such that $(G',\{E_1,\ldots,E_h\},k')$ is a yes-instance for \tempOrd\, if and only if $(G,k)$ is a yes-instance for $\mathbf{p}$-\textsc{Clique}; we will further ensure that the vertex cover number of $G'$ is bounded by a function of $k$.

We construct $G'$ as follows.  Let $P$ be a path on $k+1$ vertices, whose endpoints are denoted $s$ and $r$ respectively.  We obtain $G'$ from $P$ by adding $n\left[\binom{k}{2} + 1\right]$ new leaves $\{u_i^j: 1 \leq i \leq n, 1 \leq j \leq \binom{k}{2} + 1\}$ adjacent to $s$ and $m$ new leaves $\{w_1,\ldots,w_m\}$ adjacent to $r$.  Note that $G'$ has $k+1 + m + n\left[\binom{k}{2} + 1\right] = \mathcal{O}(m + k^2n)$ vertices.

We now define the edge subsets $\mathcal{E} = \{E_1,\ldots,E_n\}$: we have one subset $E_i$ corresponding to each vertex $v_i$ of $G$.  We set
$$E_i = \{e \in E(P)\} \cup \{su_i^j: 1 \leq j \leq \binom{k}{2} + 1\} \cup \{rw_j: e_j \text{ incident with } v_i\}.$$
To complete the construction of our instance of \tempOrd, we set $k' = |G'| - \binom{k}{2}$.  It is clear that we can construct $(G',\mathcal{E},k')$ from $(G,k)$ in polynomial time.

We begin by arguing that $s$ is the only vertex in $G'$ whose reachability set can contain more than $k'$ vertices, regardless of the choice of ordering.

\begin{claim}\label{clm:only-s}
Fix an arbitrary bijective function $t \colon \mathcal{E} \rightarrow [n]$.  Then, for any vertex $x \in V(G') \setminus \{s\}$, $|\reach_{G',\mathcal{E},t}(x)| \leq k'$.
\end{claim}
\begin{proof}[Proof of Claim \ref{clm:only-s}]
Fix $x \in V(G') \setminus \{s\}$.  It suffices to demonstrate that we can find $\binom{k}{2}$ vertices in $V(G')$ that are not in $\reach_{G',\mathcal{E},t}(x)$.

Fix $i$ to be the unique element of $[n]$ such that $E_i = t^{-1}(1)$.  We claim that at most one element of $U = \{u_i^j: 1 \leq j \leq \binom{k}{2} + 1\}$ lies in the temporal reachability set of $x$.  To see that this is the case, note that edges incident with vertices in $U$ are only active at timestep 1, and so if any such edge belongs to a strict temporal path it must be the first edge on such a path; hence there can only be a strict temporal path from some vertex $y$ to a vertex $u \in U$ if $y$ is adjacent to $u$ or $y=u$.  However, as $s$ is the only neighbour of any vertex $u$ and by assumption $x \neq u$, we can only have $u \in U \cap \reach_{G',\mathcal{E},t}(x)$ if $u = x$.  Thus $|U \cap \reach_{G',\mathcal{E},t}(x)| \leq 1$ as claimed.  Since $|U| = \binom{k}{2} + 1$, we conclude that $|\reach_{G',\mathcal{E},t}(x)| \leq k'$, as required.
\renewcommand{\qedsymbol}{$\square$ (Claim \ref{clm:only-s})}
\end{proof}

We will say that the bijective function $t \colon \mathcal{E} \rightarrow [n]$ is \emph{good for $s$} if we have $|\reach_{G',\mathcal{E},t}(s)| \leq k'$.  It follows from Claim \ref{clm:only-s} that $(G',\mathcal{E},k')$ is a yes-instance if and only if some function $t$ is good for $s$.  It therefore remains to show that there is a function $t$ which is good for $s$ if and only if $G$ contains a clique of on $k$ vertices.

To show that this is true, we first give a characterisation of the temporal reachability set of~$s$.

\begin{claim}\label{clm:s-reaches}
Fix an arbitrary bijective function $t \colon \mathcal{E} \rightarrow [n]$.  Then the only vertices of $G'$ that do not belong to $\reach_{G',\mathcal{E},t}(s)$ are vertices $w_i$ such that $e_i = v_jv_{\ell}$ and $t(E_j),t(E_{\ell}) \leq k$.
\end{claim}
\begin{proof}[Proof of Claim \ref{clm:s-reaches}]
First observe that, for any choice of $t$, $\reach_{G',\mathcal{E},t}(s)$ contains
\begin{enumerate}
\item every vertex $u_i^j$ (with $1 \leq i \leq n$ and $1 \leq j \leq \binom{k}{2} + 1$), and
\item every vertex of $P$.
\end{enumerate}
Now consider a vertex $w_i$; by definition, $w_i$ is in $\reach_{G',\mathcal{E},t}(s)$ if and only if there is a strict temporal path from $s$ to $w_i$.  There is only one possible choice of path, and the first $k$ edges on this path are active at every timestep.  Thus we have a strict temporal path from $s$ to $w_i$ if and only if the edge $rw_i$ is active at some timestep after the first $k$.  Since $rw_i$ is active only at $t(E_j)$ and $t(E_{\ell})$, where $e_i = w_jw_{\ell}$, this means that $w_i$ is in the temporal reachability set of $s$ if and only if at least one of $t(E_j)$ and $t(E_{\ell})$ is strictly greater than $k$.  Conversely, the only vertices of $G'$ that are not in $\reach_{G',\mathcal{E},t}(s)$ are the vertices $w_i$ such that $e_i = v_jv_{\ell}$ and $t(E_j),t(E_{\ell}) \leq k$, as required.
\renewcommand{\qedsymbol}{$\square$ (Claim \ref{clm:s-reaches})}
\end{proof}

Now suppose that $G$ contains a clique induced by the vertices $\{v_{i_1},\ldots,v_{i_k}\}$.  We claim that any function $t$ which maps $\{E_{i_1},\ldots,E_{i_k}\}$ to $[k]$ is good for $s$.  By Claim \ref{clm:s-reaches}, we see that the vertices of $G'$ that are not in $\reach_{G',\mathcal{E},t}(s)$ are the vertices $w_i$ such that $e_i = v_jv_{\ell}$ and $E_j,E_{\ell} \in t^{-1}([k]) = \{E_{i_1},\ldots,E_{i_k}\}$.  In other words, the vertices not in the temporal reachability set correspond to edges in $G$ which have both endpoints in the set $\{v_{i_1},\ldots,v_{i_k}\}$.  Since, by assumption, this set of vertices induces a clique, we know that there are precisely $\binom{k}{2}$ such vertices, so the reachability set misses $\binom{k}{2}$ vertices and $t$ is indeed good for $s$.

Conversely, suppose that the function $t$ is good for $s$, and set $C = \{i:t(E_i)  \leq k\}$.  We claim that $\{v_i: i \in C\}$ induces a clique in $G$.  We know by Claim \ref{clm:s-reaches} that the only vertices that do not belong to $\reach_{G',\mathcal{E},t}(s)$ are vertices $w_i$ such that $e_i = v_jv_{\ell}$ and $j,\ell \in C$.  We therefore know, since $t$ is good for $s$, that there must be $\binom{k}{2}$ unordered pairs $\{j,\ell\} \subset C$ such that $G$ contains an edge $v_iv_{\ell}$.  Since the total number of unordered pairs from $C$ is equal to $\binom{k}{2}$, it follows that there is an edge between every pair of vertices in the set $\{v_i: i \in C\}$, implying that this set of $k$ vertices does indeed induce a clique in $G$, as required.

Finally, we note that the vertices $r$ and $s$, together with every second internal vertex on the path $P$, form a vertex cover for $G'$, meaning that the vertex cover number of $G'$ is at most $k/2 + 1$.
\end{proof}

In the two preceding results, the expressive power of \tempOrd\, on highly restricted graph classes comes from the fact that, with edge-classes of size two or more, the decisions made at one location can have an effect on distant parts of the graph.  It is therefore natural to ask whether we can regain some tractability in this setting by placing structural restrictions on the edge-class interaction graph: note that, in the proof of Theorem \ref{thm:trees-hard}, although the graph $G$ we construct is a tree, the edge-class interaction graph is a clique.  

We now build on the results of Section \ref{sec:singleton-approx} to show that suitable restrictions on the edge-class interaction graph can allow the design of efficient approximation algorithms.  It remains open whether there exist efficient algorithms to solve the problem exactly when the edge-class interaction graph is sufficiently highly structured, although the hardness reduction for the singleton case in Theorem \ref{thm:sing-hard-new} already rules out a $\PTAS$ when the maximum degree of the edge-class interaction graph is bounded by a constant.  

Our first step towards an approximation algorithm for the general case is to use Lemma \ref{lma:reach-decomp} to obtain an analogous bound to that of Lemma \ref{thm:edge-col}.  Note that in the singleton case (in which each edge-class consists of a single edge) the edge-chromatic number of $G$ is precisely the chromatic number of the edge-class interaction graph.

\begin{thm}\label{thm:interaction-col-bound}
Let $(G,\mathcal{E},k)$ be an instance of \tempOrd, let $H$ be the edge-class interaction graph of $(G,\mathcal{E})$, and let $d = \max_{E' \in \mathcal{E}} \Delta((V,E'))$ be the maximum number of edges from any one element of $\mathcal{E}$ that are incident with any single vertex of $G$.  In this case there is an assignment $t: \mathcal{E} \rightarrow [|\mathcal{E}|]$ of times to edge classes such that the maximum reachability of $(G,\mathcal{E},t)$ is at most $(d+1)^{\chi(H)}$, where $\chi(H)$ is the chromatic number of $H$.
\end{thm}
 \begin{proof}
 Fix a proper vertex colouring $c:\mathcal{E} \rightarrow [\chi(H)]$ of $H$, and for each $1 \leq i \leq \chi(H)$ let $\mathcal{E}_i$ denote the subset of $\mathcal{E}$ consisting of those elements that receive colour $i$ under $c$.  By Lemma \ref{lma:reach-decomp}, it suffices to argue that for each $i$ there is an assignment $t_i$ of times to elements of $\mathcal{E}_i$ such that the maximum reachability of $(G_i = (V, \bigcup\mathcal{E}_i),\mathcal{E}_i,t_i)$ is at most $d + 1$.

 We note that, by definition of the edge-class interaction graph and a proper colouring, any pair of incident edges in $G_i$ must belong to the same element of $\mathcal{E}$.  Thus, in any connected component of $G_i$, all edges will be assigned the same timestep in any ordering, meaning that no temporal path can contain more than one edge.  It follows that, for each $i$, in $(G_i,\mathcal{E}_i,t_i)$, each vertex reaches only itself and its neighbours.  The required bound on the maximum reachability now follows immediately from the bound on the degree in each element of $\mathcal{E}$, since all edges of $G_i$ incident with a single vertex must belong to a single element of $\mathcal{E}$.
 \end{proof}

Theorem \ref{thm:interaction-col-bound} immediately gives rise to an efficient constant factor approximation algorithm for \tempOrd, whenever it is possible to compute the chromatic number of the edge-class interaction graph (or a constant-factor approximation to this value) efficiently, provided that the maximum degree of any edge class is bounded.  Note that this second condition will certainly be satisfied if the maximum degree of the input graph is bounded.  We obtain the following corollary by recalling that (1) every graph $H$ admits a proper $(\Delta(H) + 1)$-colouring, which can be constructed greedily, and (2) it is possible to construct a proper 2-colouring, if any exists, in linear time.

\begin{cor}
Let $(G = (V, E),\mathcal{E},k)$ be an instance of \tempOrd, suppose that $\max_{E' \in \mathcal{E}} \Delta_{E'} \leq d$, where $\Delta_{E'}$ denotes the maximum degree of the graph $(V(G),E')$, and let $H$ be the edge-class interaction graph of $G$.  Then we can compute, in polynomial time, an approximation to the optimisation version of \tempOrd\, with approximation ratio $d + 1$ if $H$ is bipartite, and approximation ratio $(d+1)^{\Delta(H)}$ otherwise.
\end{cor}

\section{Conclusions and future work}

We have shown \tempOrd\, is extremely difficult to solve exactly: it remains intractable even in the special case of pairwise disjoint singleton edge-classes, and even two highly restricted cases (DAGs, and trees parameterised by vertex cover number) which are almost trivial in this singleton case become intractable as soon as larger edge-classes are allowed.  We have already highlighted one key open question that remains regarding the complexity of solving the problem exactly, namely determining the existence or otherwise of a polynomial-time algorithm for trees of unbounded degree.  It would also be interesting, based on the observation that networks of epidemiological interest often have few disjoint directed cycles, to investigate whether our positive result concerning DAGs in the singleton case can be generalised to directed graphs with bounded size feedback vertex~sets.

Given the strength of our hardness results, it is natural to seek approximation algorithms for the optimisation version of the problem, and we have shown that there exist efficient algorithms to compute constant-factor approximations on graphs of bounded degree either in the singleton case or when the edge-class interaction graph is bipartite.  However, the approximation factor of these algorithm increases very fast with the degree of the input graph and, while we have ruled out the existence of a $\PTAS$ even for the singleton version on bounded degree graphs unless $\P=\NP$, it is still possible that there exists a polynomial-time $c$-approximation algorithm for some $c > \frac{31}{30}$ that does not depend on the degree of the graph.

In spite of the many intractability results for the problems considered so far, the overall goal remains to model real-world reordering problems of this form, and to this end it would be relevant to consider a generalisation of \tempOrd\, in which each edge-class has a list of permitted times.  This would correspond to practical restrictions on when each individual event can be scheduled: for example, a particular event might have to be scheduled for a specific day of the week or at a certain time of year (while the ``Spring Bull Sale'' can perhaps be rescheduled within a window of a several weeks, it is probably not acceptable to move it to October).  We believe that our FPT algorithm for trees (parameterised by the maximum permitted reachability) can be adapted to deal with a list version in the singleton case, but it is not even clear whether the singleton version can be solved efficiently on DAGs when arbitrary lists of permitted times are allowed.  A further generalisation would be to associate different costs with the possible times for each edge class (to model the different costs that might be associated with rescheduling events to different times), and seek to minimise the maximum reachability subject to a given budget constraint.

A number of other variations would also be of practical interest.  For certain applications, it might also be relevant to consider the problem of minimising the \emph{average} cardinality of the reachability set over all vertices of the graph, or indeed the expected size of the reachability set (perhaps given some distribution over starting vertices) in a probabilistic model where each edge has an associated transmission probability.  Additionally, previous work on temporal graphs has addressed a notion of $(\alpha,\beta)$-reachability, in which the timesteps at which consecutive edges in a temporal path are active must differ by at least $\alpha$ and at most $\beta$ \cite{temporalReview,temp-edge-del}; this is a more realistic model for the spread of a disease, as individuals are not instantaneously infectious when infected, and do not remain infectious indefinitely.  It would be very interesting to investigate this problem in the reordering context introduced~here.

  \bibliography{ESA_temporal_refs}

\end{document}